\documentclass[pra,twocolumn,floatfix,notitlepage,groupedaddress,longbibliography]{revtex4-1}

\usepackage[colorlinks=true,citecolor=black,linkcolor=black]{hyperref}
\usepackage{mathtools}

\usepackage{amssymb,amsmath,amstext,amsthm}
\usepackage{graphicx}
\usepackage{epstopdf}
\usepackage{color}
\usepackage{bm}
\usepackage{appendix}
\usepackage[T1]{fontenc}
\usepackage{bbm}
\usepackage{latexsym}

\usepackage{float}
\usepackage{verbatim}
\usepackage{lipsum}
\usepackage{braket}
\usepackage{ulem}

\newcommand{\del}[0]{\partial}

\newtheorem{theorem}{Theorem}

\newtheorem{proposition}{Proposition}
\begin{document}

\title{Asymptotic optimality of twist-untwist protocols for Heisenberg scaling in atom-based sensing}

\author{T.J.\,Volkoff}
\affiliation{Theoretical Division, Los Alamos National Laboratory, Los Alamos, NM, USA.}
\author{Michael J. Martin}
\affiliation{Materials Physics and Applications Division, Los Alamos National Laboratory, Los Alamos, NM, USA.}

\begin{abstract}
Twist-untwist protocols for quantum metrology consist of a serial application of: 1. unitary nonlinear dynamics (e.g., spin squeezing or Kerr nonlinearity), 2. parameterized dynamics $U(\phi)$ (e.g., a collective rotation or phase space displacement), 3. time reversed application of step 1. Such protocols are known to produce states that allow Heisenberg scaling for experimentally accessible estimators of $\phi$ even when the nonlinearities are applied for times much shorter than required to produce Schr\"{o}dinger cat states. In this work, we prove that, asymptotically in the number of particles, twist-untwist protocols provide the lowest estimation error among quantum metrology protocols that utilize two calls to a weakly nonlinear evolution and a readout involving only first and second moments of a total spin operator $\vec{n}\cdot \vec{J}$. We consider the following physical settings: all-to-all interactions generated by one-axis twisting $J_{z}^{2}$ (e.g., interacting Bose gases), constant finite range spin-spin interactions of distinguishable or bosonic atoms (e.g., trapped ions or Rydberg atoms, or lattice bosons). In these settings, we further show that  the optimal twist-untwist protocols asymptotically achieve 85\%  and 92\% of the respective quantum Cram\'{e}r-Rao bounds. We show that the error of a twist-untwist protocol can be decreased by a factor of $L$ without an increase in the noise of the spin measurement if the twist-untwist protocol can be noiselessly iterated as an $L$ layer quantum alternating operator ansatz.
\end{abstract}
\maketitle

\section{Introduction}\label{sec:intro}
Advances in experimental implementations of squeezing-enhanced quantum metrology protocols \cite{PhysRevA.93.013851,thomp,PhysRevLett.122.030501,PhysRevLett.122.223203,qpm,vuletic} emphasize the fact that non-classicality and many-body entanglement are resources for near-term quantum technologies. In particular, quantum circuits consisting of  alternating squeezing and unsqueezing operations have allowed to amplify signals in systems such as microwave photons \cite{PhysRevLett.120.040505}, cooled mechanical oscillators \cite{winey}, and atomic ensembles (\textit{twist-untwist} protocols) \cite{PhysRevLett.116.053601,PhysRevLett.117.013001}. The possibility of integrating these circuits as modules of variational quantum sensing algorithms further suggests that squeezing-enhanced sensing could be utilized in near-term quantum computers  \cite{PhysRevLett.123.260505}.  

Control of two-body elastic scattering allows generation of entanglement by the phenomenon of spin squeezing \cite{PhysRevLett.86.4431,PhysRevA.47.5138,MA201189}.
Specifically, in the case of spatially interacting atoms as in a dilute quantum gas, combining tunable interatomic interactions such as Feshbach resonances \cite{PhysRevLett.92.160406} with optical trap manipulations can be used modify the strength and type of spin squeezing \cite{folling,PhysRevA.94.042327}. Alternatively, spin squeezing of internal atomic states in  hybrid atomic-optical systems can be produced by controllable Rydberg interactions via laser dressing \cite{bieder,PhysRevA.82.033412,PhysRevLett.104.195302,PhysRevLett.105.160404,mmd,Borregaard_2017}. Although spin squeezing can be analyzed by emphasizing analogies with continuous-variable quadrature squeezing, it is a many-body quantum effect, in the sense that there are spin squeezing quantifiers that imply particle entanglement of a given many-body state \cite{PhysRevA.69.052327}. In constrast, continuous-variable squeezing is not sufficient for entanglement despite it being the cause of several non-classical optical phenomena.

The relation between spin squeezing and quantum metrology is that spin squeezed probe states allow estimate frequencies with error that scales as $1/N^{\beta}$ with $\beta >1$ \cite{wine1,wine2,PhysRevA.47.5138}. In a two-mode bosonic system of $N$ particles, or a system of $N$ qubits, the frequency parameter couples to a population difference operator $J_{z}$ which has operator norm $N/2$. The error of the frequency estimator is bounded below by the reciprocal of the quantum Fisher information (QFI) according to the quantum Cram\'{e}r-Rao inequality (QCRI) \cite{holevo}. When the system is unentangled, a QFI of $N^{\beta}$ with $\beta =1$ is the largest possible; the resulting QCRI is called the standard quantum limit. However, for a system prepared in a Greenberger-Horne-Zeilinger (GHZ) state, a QFI of $N^{\beta}$ with $\beta =2$ is the largest possible; the resulting QCRI is called the Heisenberg limit \cite{PhysRevLett.96.010401}. An estimator with error scaling as $O(N^{-2})$ is said to exhibit Heisenberg scaling. By SU(2) symmetry, analogous statements hold when the parameter to be estimated couples to a general spin operator $\vec{n}\cdot \vec{J}$ with $\Vert \vec{n}\Vert=1$.

 The fact that squeezing-unsqueezing protocols can enhance measurement sensitivity for quantum estimation protocols beyond the standard quantum limit was first elucidated in the context of photonic Mach-Zehnder and four-wave-mixing interferometers \cite{PhysRevA.33.4033}. By replacing the non-passive elements of these photonic interferometers by entanglement-generating atomic interactions such as one-axis or two-axis twisting, one is led to protocols that achieve analogous scaling of the sensitivity with respect to the number of atoms.  Specifically, in the context of phase sensing with atomic ensembles, a single layer one-axis twist-untwist protocol is defined by the parameterized $N$ particle quantum state $\ket{\psi_{\phi}}$ where
\begin{equation}
\ket{\psi_{\phi}}=  e^{i\chi t J_{z}^{2}}e^{-i\phi J_{y}}e^{-i\chi t J_{z}^{2}}  \ket{+}^{\otimes N}.
\label{eqn:iii}
\end{equation}
In (\ref{eqn:iii}), the spin operators satisfy the $\mathfrak{su}(2)$ relation $[J_{i},J_{j}]=i\epsilon_{ijk}J_{k}$, and $\ket{+}^{\otimes N}$ is the maximal $J_{x}$ eigenvector in the spin-${N\over 2}$ representation of $SU(2)$.
\begin{figure*}[t!]
\includegraphics[scale=0.35]{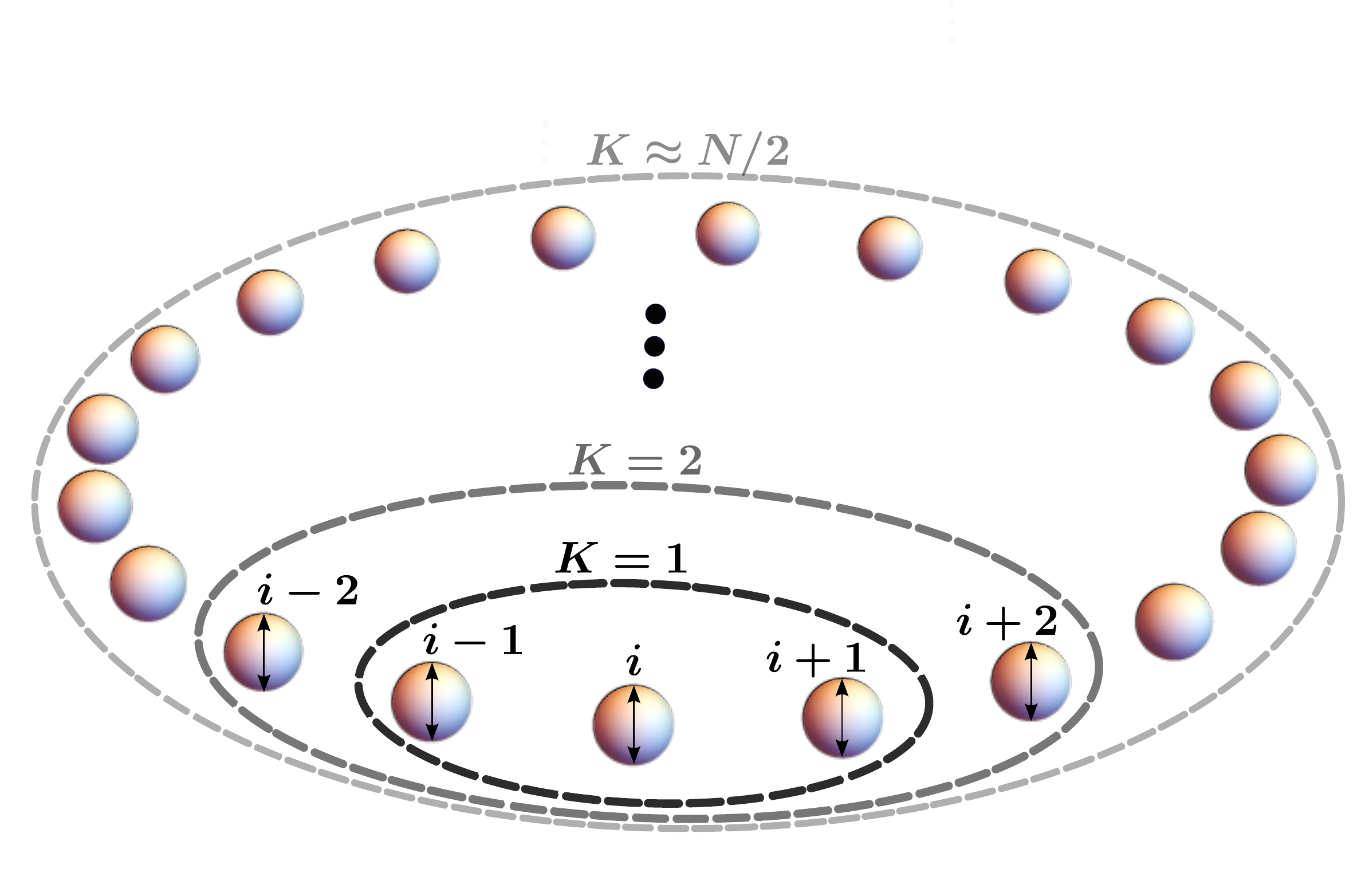}
    \caption{Finite range interactions (parameterized by $K$) between individually-trapped atoms in a periodic 1-D array. We show the optimality of the twist-untwist protocols first for full-range interactions ($K\approx N/2$, Section \ref{sec:ao}), then for finite range interactions ($K<N/2$, Section \ref{sec:fr}). In the case of a translationally-invariant Bose-symmetrized system (Section \ref{sec:gen}), we find that finite range interactions rescale the interaction strength of an all-to-all Hamiltonian that maps onto a global spin-squeezing operator.
}
    \label{fig:schem}
\end{figure*}
A major advantage of the probe state (\ref{eqn:iii}) compared to, e.g., the one-axis twisting probe state $e^{-i\chi t J_{z}^{2}}\ket{+}^{\otimes N}$, is that large values of the effective interaction time $\chi t$ are not required in order to achieve Heisenberg scaling for the method-of-moments estimation of the phase $\phi$  \cite{PhysRevLett.116.053601}. In particular, generation of the state that maximizes the QFI appearing in the QCRI for estimation of a $y$-rotation, viz., an equal amplitude Schr\"{o}dinger cat state of the minimal and maximal $J_{y}$ eigenvectors, is not required. Further, the untwist operation in (\ref{eqn:iii}) serves to realign the probe state so that a $J_{y}$ measurement gives Heisenberg scaling for all $N$. For comparison, note that, e.g., at $\chi t=\pi/2$, the alignment of the GHZ state produced by one-axis twisting dynamics depends on the parity of $N$. Therefore, the axis of the optimal total spin measurement depends on the parity of $N$. This parity is challenging to keep constant over a sequence of experimental runs. However, the question of the conditions under which protocol (\ref{eqn:iii}) is an optimal strategy for achieving Heisenberg scaling for estimation of $\phi$ remains open. 

In this work, we demonstrate the optimality of probe states of the form (\ref{eqn:iii}) in several settings relevant to quantum sensing with atomic systems. Unlike protocols based on noise-robust interaction-based readouts utilizing spin-resolved measurement statistics \cite{PhysRevA.97.053618,PhysRevA.98.030303,PhysRevLett.119.193601} or Loschmidt echo sequences that utilize information about the parameterized fidelity of the input and output states \cite{PhysRevA.94.010102}, we do not seek an optimal strategy for saturating a QCRI. Rather, our main results in Theorems \ref{th:1} and \ref{th:2} demonstrate the asymptotic ($N\rightarrow \infty$) optimality of the twist-untwist probe state (\ref{eqn:iii}) relative to a total spin measurement and its associated empirical error quantifier based on first and second moments  (see Eq.(\ref{eqn:prec})).  The fact that the ratio of the reciprocal QFI to the  asymptotic empirical error of the twist-untwist protocol is $\ge 0.85$ in all settings considered shows that the fundamental limit on the precision of an estimate of $\phi$ is lower than the empirical error only by a small multiplicative factor, and indicates the utility of this simple strategy for near-optimal quantum estimation.

In more detail, our main results show that the probe state (\ref{eqn:iii}) in the limit $N\rightarrow \infty$ achieves the minimal error possible among all protocols that apply a weak one-axis twisting before and after the rotation parameter  (Section \ref{sec:ao}).  In Section \ref{sec:layer} and Section \ref{sec:fr}, respectively, we obtain analogous results for a multilayer improvement of the probe state (\ref{eqn:iii}), and its implementation in systems with uniform, finite range atom-atom interactions. Fig. \ref{fig:schem} illustrates a schematic implementation of the interactions we consider in a periodic 1-D lattice of atoms. In our analyses of optimality, the principal physical constraint is the requirement of an asymptotically vanishing interaction time $\chi t \rightarrow 0$ as $N\rightarrow \infty$. This constraint is motivated by the fidelity losses encountered when generating coherence and entanglement of a many-atom system by spin-squeezing over long times.  Generation of spin squeezing by one-axis twisting is achievable in a 2-D system of trapped ions with  $N=O(100)$ \cite{boll}, and spin squeezing of approximately 20 dB has been achieved in the clock states of $^{87}$Rb \cite{kase,RevModPhys.90.035005}. The twist-untwist protocol (\ref{eqn:iii}) and readout of a spin direction observable achieving an estimation precision below the standard quantum limit has been implemented in a system of $N=O(100)$ $^{171}$Yb atoms \cite{vuletic}.
 
 \section{Background}
The setting of the quantum metrology problem at hand consists of: 1. preparation of a spin-${N\over 2}$ coherent state of $N$ two-level atoms, 2. application of alternating twist-untwist interactions sandwiching calls to the parameter $\phi$, and 3. measurement of $J_{y}$.  We define the empirical error as \begin{equation}(\Delta \phi)^{2}:= {\text{Var}_{\ket{\psi_{\phi}}}J_{y}\over \left( \del_{\phi} \bra{\psi_{\phi}}J_{y}\ket{\psi_{\phi}}\right)^{2}}.\label{eqn:prec}\end{equation}
Mathematically, this quantity is the asymptotic error-per-shot when using $J_{y}$ for method-of-moments estimation of $\phi$ \cite{RevModPhys.90.035005,PhysRevLett.122.090503}. 
When applied to the parameterized state (\ref{eqn:iii}), the  empirical error (\ref{eqn:prec})  exhibits Heisenberg scaling  $O(N^{-2})$, similar to the quantum Cram\'{e}r-Rao bound for estimation of $\phi$ which coordinatizes the a quantum state manifold (\ref{eqn:iii}) (see Proposition \ref{prop:aaa} for a rigorous statement). Note that achieving Heisenberg scaling exactly equal to the QFI requires implementation of an optimal measurement which is not a total spin observable. Since the error (\ref{eqn:prec}) is invariant under $\chi \rightarrow -\chi$,  the probe state (\ref{eqn:iii}) can be defined with $\chi>0$ without loss of generality. As discussed in Section \ref{sec:ao}, and as shown previously \cite{PhysRevLett.116.053601}, Heisenberg scaling of (\ref{eqn:prec}) occurs even for a weak interaction time $\chi t$, e.g., $\chi t=O( N^{-1/2})$. A physical explanation of this fact is that the initial one-axis twisting drives the coherent state $\ket{+}^{\otimes N}$ toward a Schr\"{o}dinger cat state of  according to the Yurke-Stoler dynamics driven by the one-axis twisting. For small $\chi t$, this process actually creates a pseudo-cat state, which is more sensitive to rotation than the initial spin coherent state. The untwisting acts to amplify the signal due to the rotation (the denominator of (\ref{eqn:prec}) while keeping the variance constant. Note that the twist-untwist protocol does not return the initial spin coherent state to the manifold of spin coherent states.

The fact that the nonlinearity of an interaction can compensate for weak interaction strength to achieve Heisenberg scaling in quantum sensing can also be observed for continuous variable displacement sensing. The continuous variable analogue of the twist-untwist protocol is given by applications of the Kerr nonlinearity with opposite signs: $\ket{\psi_{\phi}}=e^{i\chi t (a^{*}a)^{2}}D(\phi)e^{-i\chi t (a^{*}a)^{2}}\ket{\alpha}$ where $D(\phi)=e^{-i\phi p}$ is a unitary displacement operator and $\ket{\alpha}$ is a Heisenberg-Weyl coherent state with $\text{Im}\, \alpha =0$. One finds that for a homodyne readout of the $p$-quadrature
\begin{align}
    (\Delta \phi)^{2}\big\vert_{\phi=0}&:= {\text{Var}_{\ket{\psi_{\phi}}}p\over ({d\over d\phi}\langle \psi_{\phi}\vert p\vert \psi_{\phi} \rangle )^{2}}\Big\vert_{\phi=0}\nonumber \\&= {1\over 2\alpha^{4}e^{-4\alpha^{2}\sin^{2}\chi t}\sin^{2}\left( \alpha^{2}\sin(2\chi t)+1\right)} \label{eqn:cvtw}
\end{align}
which, for large $\alpha$ has an approximate minimum at $\chi t={\pi -2\over 4(\alpha^{2} +1)}$, an interaction time at which (\ref{eqn:cvtw}) scales as the inverse square of the intensity $\alpha^{2}$.

\section{Asymptotic optimality of twist-untwist protocols\label{sec:ao}}

In this section, we analyze a class of interferometry protocols involving two calls to weak one-axis twisting dynamics. The system consists of $N$ bosonic atoms distributed among two orthogonal single-particle modes, e.g., two matter waves of a weakly interacting Bose gas with different momenta. The parameter $\phi$ to be sensed determines the transmissivity of the matter-wave beamsplitter. One can view the spin operators $J_{i}$ by their Schwinger boson realization. Equivalent schemes can be engineered in cavity-QED or Rydberg-dressed atom systems in which the orthogonal single particle modes are internal atomic states, and sensing the interferometer parameter is equivalent to Ramsey spectroscopy up to a $\pi/2$ rotation  \cite{PhysRevLett.116.053601}.

Our main result in Theorem \ref{th:1} shows that when using the first and second moments of the total $J_{y}$ spin to quantify the error as in (\ref{eqn:prec}), and for two calls to a weak one-axis twisting evolution separated by the spin rotation to be sensed, the twist-untwist protocol (\ref{eqn:iii}) gives an optimal probe state in the limit of large $N$. However, it is useful to first understand how well, given the twist-untwist protocol (\ref{eqn:iii}), the error (\ref{eqn:prec}) performs when compared to the ultimate precision obtainable by an optimal unbiased estimator in one-shot quantum estimation theory.

For this, we recall that Heisenberg scaling of an unbiased estimator of $\phi$ is possible when the QFI appearing in the QCRI scales as $O(N^{2})$. Therefore, we first provide a basic, but rigorous, statement that relates the optimal scaling of the method-of-moments error (\ref{eqn:prec}) for twist-untwist protocols to $O(N^{2})$ scaling of the QFI of the twist-untwist protocol (\ref{eqn:iii}).

\begin{proposition}
The minimum of (\ref{eqn:prec}) with respect to twist-untwist protocol (\ref{eqn:iii}) occurs at $\chi t = \tan^{-1}{1\over \sqrt{N-2}}$ with minimum value asymptotically given by ${e\over N^{2}}$ as $N\rightarrow \infty$. With $\mathrm{QFI}(\psi_{\phi})$ defined as the quantum Fisher information for the protocol (\ref{eqn:iii}), $\mathrm{QFI}(\psi_{\phi})\sim \left( {e^{2}-1\over 2e^{2}}\right)N^{2}$ at $\chi t =\tan^{-1}{1\over \sqrt{N-2}}$ and the function
\begin{equation}
    f(\chi t):= {\mathrm{QFI}(\psi_{\phi})^{-1}\over (\Delta \phi)^{2}\vert_{\phi=0}}
\end{equation}
satisfies
$f(\tan^{-1}{1\over \sqrt{N-2}})\sim {2\over e-e^{-1}}$. Further, $f(\chi t)\le 1$ and the maximum value 1 is asymptotically attained when $\chi t$ is $O({1\over N})$.
\label{prop:aaa}
\end{proposition}
\begin{proof}
The value $\tan^{-1}{1\over \sqrt{N-2}}$ for the critical interaction time is proven in Ref. \cite{PhysRevLett.116.053601}. Note that the numerator of $f$ is the lower bound appearing in the quantum Cram\'{e}r-Rao inequality (QCRI), so $f=1$ implies that the measurement saturating the QCRI for the protocol $\ket{\psi_{\phi}}$ has the same error as the $J_{y}$ measurement defining (\ref{eqn:prec}). The fact that $f\le 1$ follows from the fact that $((\Delta \phi)^{2}\vert_{\phi =0})^{-1}$ is at most the classical Fisher information with respect to the $J_{y}$ measurement at $\ket{\psi_{\phi=0}}$ \cite{ps}, and the existence of a measurement for which the classical Fisher information saturates the QFI \cite{bc}.  From the symmetric logarithmic derivative for the state manifold $\ket{\psi_{\phi}}$ (see (\ref{eqn:sld})), it follows that \begin{align}\text{QFI}(\psi_{\phi})&=4\text{Var}_{e^{\pm i\chi t J_{z}^{2}}\ket{+}^{\otimes N}}J_{y} \nonumber \\
&= {1\over 2}\left( N^{2}+N - N(N-1)\cos^{N-2}2\chi t \right).
\label{eqn:qfiqfi}
\end{align}
Note that $\text{QFI}(\psi_{\phi})$ is independent of $\phi$. At $\chi t=\tan^{-1}{1\over \sqrt{N-2}}$ one finds that
\begin{align}
&{}={1\over 2}\left( N^{2}+N - N(N-1)\cos^{N-2}\left( 2\tan^{-1}{1\over \sqrt{N-2}}\right) \right)\nonumber \\
&=  {1\over 2}\left( N^{2}+N - N(N-1)\left( {2\over 1+{1\over N-2}}-1\right)^{N-2}\right)\nonumber \\
&\sim {1\over 2}\left( N^{2}+N - N(N-1)\cdot {1\over e^{2}}\right)\nonumber \\
&\sim {e^{2}-1\over 2e^{2}}N^{2}.
\end{align} 
A closed formula for the empirical error (\ref{eqn:prec}) evaluated for the more general protocol of the form
\begin{equation}
    \ket{\psi_{\phi}}=e^{ia_{2}J_{z}^{2}}e^{-i\phi J_{y}}e^{ia_{1}J_{z}^{2}}  \ket{+}^{\otimes N}
    \label{eqn:twoparam}
\end{equation} is given by:
\begin{small}
\begin{align}
    &{}(\Delta \phi)^{2}\big\vert_{\phi =0}\nonumber \\
    &= {2(N+1)-2(N-1)\cos^{N-2}(2(a_{1}+a_{2})) \over N(N-1)^{2}\sin^{2}a_{2}\left( \cos^{N-2}a_{2} + \cos^{N-2}(2a_{1}+a_{2}) \right)^{2}}.
    \label{eqn:anal2}
\end{align}
\end{small}
The function $f(\chi t)$ defined with respect to twist-untwist protocol (\ref{eqn:iii}) is calculated from (\ref{eqn:anal2}) and (\ref{eqn:qfiqfi}), giving the final result
\begin{equation}
    f(\chi t)={2N(N-1)^{2}\sin^{2}\chi t\cos^{2N-4}\chi t\over N^{2}\left( 1-\cos^{N-2}2\chi t \right) + N\left( 1+\cos^{N-2}2\chi t \right)}.
    \label{eqn:gfgf}
\end{equation}
Using $\tan^{-1}{1\over \sqrt{N-2}}\sim {1\over \sqrt{N}}$, $\cos^{N}{2\over \sqrt{N}}\sim e^{-2}$, $\cos^{2N-4}{1\over \sqrt{N}}\sim e^{-1}$, and $\sin^{2}{1\over \sqrt{N-2}}\sim {1\over N}$ gives the asymptotic result $f(\tan^{-1}{1\over \sqrt{N-2}})\sim {2\over e-e^{-1}}$, which implies that at the  interaction time that minimizes (\ref{eqn:prec}), the lowest possible achievable error as expressed by the QCRI associated with (\ref{eqn:qfiqfi}) is only about 15\% lower than (\ref{eqn:prec}) in the limit $N\rightarrow \infty$. Finally, we show the existence of an interaction time scaling for which $\lim_{N\rightarrow \infty}f(\chi t)=1$. Simply consider $\chi t= {1\over N}$ and use the asymptotics $\cos^{N}{1\over N}\sim \cos^{2N}{1\over N}\sim 1$ and $\sin^{2}{1\over N} \sim {1\over N^{2}}$ to get $\lim_{N\rightarrow \infty}f({1\over N})=1$. However, for $\chi t =O( {1\over N})$, the QFI is $O(N)$, so a twist untwist protocol for such short interaction times is not useful for sensing $\phi$ below the standard quantum limit.
\end{proof}

The function $f$ is plotted in Fig.\ref{fig:ooo2} for $N=10^{3}$. For $\chi t > \tan^{-1}{1\over \sqrt{N-2}}$, the one-axis twisting probe state $e^{\pm i \chi t J_{z}^{2}}\ket{+}^{\otimes N}$ has QFI scaling as $O(N^{2})$, which cannot be obtained with the twist-untwist protocol (\ref{eqn:iii}). However, for $\chi t \le \tan^{-1}{1\over \sqrt{N-2}}$, the twist-untwist protocol with error (\ref{eqn:prec}) scales similarly to the optimal error achievable with the one-axis twisting probe, with both quantities scaling as $O(N^{-2})$ when $\chi t \approx \tan^{-1}{1\over \sqrt{N-2}}$.

Although Proposition \ref{prop:aaa} suggests how to interpret the optimal $N^{-2}$ scaling of (\ref{eqn:prec}) for the twist-untwist protocol  (\ref{eqn:iii}), it remains unclear whether similar protocols involving one-axis twisting before and after the rotation would be able to achieve the same scaling. Therefore, we now consider protocols of the form (\ref{eqn:twoparam})
with $a_{j}\in \mathbb{R}$. Numerical optimization of the empirical error (\ref{eqn:prec}) and effects of dephasing noise for such protocols were considered in \cite{Schulte2020ramsey}. Note that calculation of the phase estimate or the empirical error (\ref{eqn:prec}) do not require access to the full probability distribution obtained from measurement of an observable. Different phase estimation schemes can lead to different assessments of the optimality of protocols within the class (\ref{eqn:twoparam}). For example, Loschmidt echo protocols for estimation of $\phi$ in (\ref{eqn:twoparam}), which require a measurement that allows to estimate the probability $\vert \left( \ket{ \psi_{\phi}}, \ket{+}^{\otimes N}\right)\vert^{2}$, are known to saturate the QCRI for $a_{2}=-a_{1}$ \cite{PhysRevA.94.010102}. In contrast, numerical evidence suggests that protocol (\ref{eqn:twoparam}) with $\vert a_{2}\vert \neq \vert a_{1}\vert$ can allow a greater classical Fisher information than the case of  $a_{2}=-a_{1}$ when employing noisy spin-resolving measurements   \cite{PhysRevLett.119.193601}.

Note that the formula (\ref{eqn:anal2}) demands that we refine the parameter space of (\ref{eqn:twoparam}) so as to have a well-defined sensing protocol. In particular, we restrict to $a_{2}\in (-\pi/2,0)\cup (0,\pi/2)$ and $a_{1}\in (-\pi/2,0)$ without loss of generality. In experimental implementations of (\ref{eqn:twoparam}), the range of available interaction times $a_{j}$ will depend on $N$, due to decoherence.

We now aim to show that when the interaction times $a_{1}$ and $ a_{2}$ are $N$-dependent functions that go to zero as $N\rightarrow \infty$, (\ref{eqn:anal2}) is asymptotically minimized for $a_{2}=-a_{1}$, i.e., at a point at which (\ref{eqn:twoparam}) defines a twist-untwist protocol. The key observation is that for fixed $a_{1}$, (\ref{eqn:anal2}) has an asymptotic extremum at $a_{2}=-a_{1}$. This is shown in Theorem \ref{th:1} below.

\begin{theorem}
Let $(\Delta \phi)^{2}\big\vert_{\phi=0}$ be defined with respect to $\ket{\psi_{\phi}}$ as in (\ref{eqn:twoparam}) and let $a_{j}={c_{j}\over N^{\alpha}}$ , $\alpha \ge 1/2$, for nonzero constants $c_{j}$, $j=1,2$. Then, in the limit $N\rightarrow \infty$ the unique minimum of $N^{2}(\Delta \phi)^{2}\big\vert_{\phi=0}$ occurs at $c_{2}=-c_{1}=1$.
\label{th:1}
\end{theorem}
\begin{proof}
 Call $f(a_{1},a_{2}) :=N^{2}(\Delta \phi)^{2}\big\vert_{\phi=0}$ and note that we restrict to $a_{1}<0$. The factor of $N^{2}$ in the definition of $f$ is so that the $N\rightarrow \infty$ limit of $\nabla f$ is not zero pointwise; without this factor, $f$ is asymptotically zero. For example, with $c:=\tan^{-1}{1\over\sqrt{N-2}}$, it follows that $\lim_{N\rightarrow \infty}f(-c,c)=e$. Consider first $\alpha>1/2$. From (\ref{eqn:anal2}) and the fact that $\cos^{N-2}{x\over N^{\alpha}}\sim 1$ for $\alpha >1/2$, it follows that $f(a_{1},a_{2})\sim {N^{2\alpha -1}\over c_{2}^{2}}$. Therefore, for $\alpha >1/2$, $f$ is arbitrarily large as $N$ increases. For $\alpha=1/2$, note that if $c_{2}\neq -c_{1}$, then the numerator of $f$ scales as $N^{2}$ for large $N$ whereas the denominator scales as $N$ for large $N$. Therefore, if $\alpha = 1/2$ and $c_{2}\neq -c_{1}$, then $f$ is arbitrarily large as $N$ increases. However, if $\alpha=1/2$ and $c_{2}= -c_{1}$, then $f\sim {1\over c_{2}^{2}e^{-c_{2}^{2}}}$, which does not scale with $N$. The asymptotic is minimized when $c_{2}=1$. 
\end{proof}
A straightforward corollary of Theorem \ref{th:1} is that the twist-untwist protocol (\ref{eqn:iii}) with $\chi t={1\over \sqrt{N}}$ is asymptotically optimal among protocols of the form (\ref{eqn:twoparam}) restricted to weak nonlinearities $a_{j}={c_{j}\over N^{\alpha}}$, $\alpha \le 1/2$. One can see this from Theorem \ref{th:1} by noting that in the cases: 1. $\alpha >1/2$ and 2. $\alpha=1/2$ and $c_{2}\neq -c_{1}$, the empirical error satisfies $(\Delta \phi)^{-2}\big\vert_{\phi=0} = o(N^{2})$. In the remaning case, viz., $\alpha=1/2$ and $c_{2}= -c_{1}$, $(\Delta \phi)^{-2}\big\vert_{\phi=0} = O(N^{2})$. Therefore, under the weak nonlinearity constraint expressed in Theorem \ref{th:1}, the asymptotically optimal parameters of (\ref{eqn:twoparam}) exhibit $N^{-1/2}$ decay. This decay is modified when the interaction has finite range. We discuss this further in Section \ref{sec:fr}.

\begin{figure*}[t!]
\includegraphics[scale=0.6]{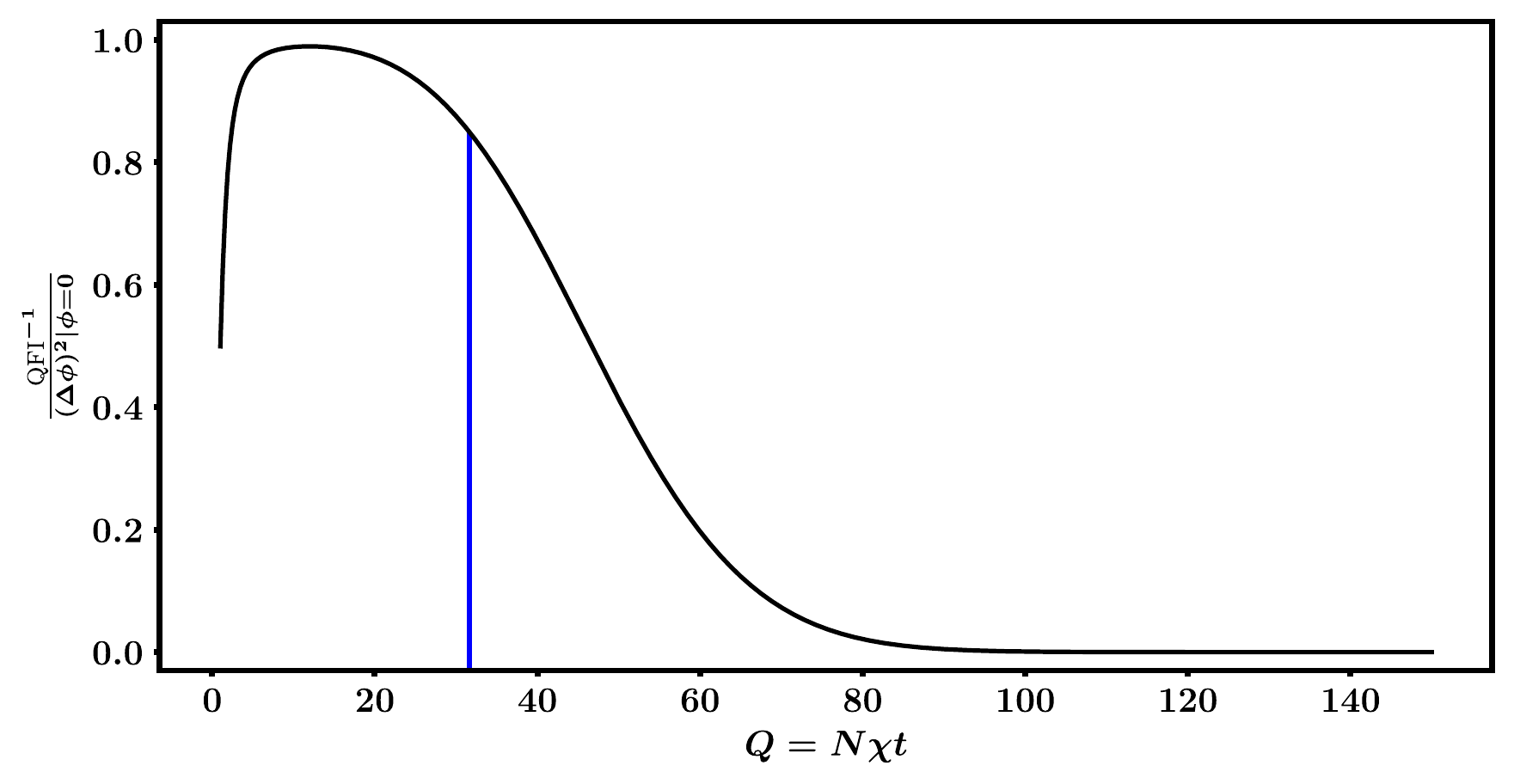}
    \caption{The function $f$ in (\ref{eqn:gfgf}) for $N=10^{3}$. The blue line is at $\tan^{-1}{1\over \sqrt{N-2}}$, the critical point for the minimum of (\ref{eqn:prec}) for the protocol (\ref{eqn:iii}).}
    \label{fig:ooo2}
\end{figure*}

We note that applying a counter-rotation at the end of a twist-untwist protocol, i.e., forming the protocol $\ket{\zeta_{\phi}}:=e^{i\phi J_{y}}e^{i\chi t J_{z}^{2}}e^{-i\phi J_{y}}e^{-i\chi t J_{z}^{2}}\vert +\rangle^{\otimes N}$, does not change the value of  (\ref{eqn:prec}). However, the QFI is no longer given by (\ref{eqn:qfiqfi}) and is rather given by
\begin{align}
\text{QFI}(\zeta_{\phi})\vert_{\phi=0}&:= \Vert \mathcal{L}_{\phi=0}\ket{+}^{\otimes N}\Vert^{2} \nonumber \\
&=  \text{QFI}(\psi_{\phi}) + N-2N\cos^{N-1}\chi t
\end{align}
where 
\begin{equation}
\mathcal{L}_{\phi=0}=\left( 2iJ_{y} -2ie^{i\chi t J_{z}^{2}}J_{y}e^{-i\chi t J_{z}^{2}} \right)\ket{+}\bra{+}^{\otimes N} + h.c.
\end{equation}
is the symmetric logarithmic derivative. This does not change the $N\rightarrow \infty$ asymptotic result for $f(\chi t)$ in Proposition \ref{prop:aaa} at the optimal interaction time $\chi t = \tan^{-1}{1\over \sqrt{N-2}}$.

\section{Layered twist-untwist protocols\label{sec:layer}}

An $L$ layer twist-untwist protocol can be defined by the parameterized state
\begin{equation}
\ket{\psi_{\phi}^{(L)}}=e^{i\phi J_{y}} \left( e^{-i\phi J_{y}} e^{i\chi t J_{z}^{2}}e^{-i\phi J_{y}}e^{-i\chi t J_{z}^{2}}  \right)^{L}\ket{+}^{\otimes N}.
\label{eqn:lla}
\end{equation}
 The state (\ref{eqn:lla}) is motivated by the consideration of a twist-untwist layer $e^{i\chi t C}e^{-i\phi J_{y}}e^{-i\chi t C}$ as a module for interferometry. Similarly structured circuits appear in asymptotically optimal variational quantum algorithms for quantum unstructured search \cite{rieffel}.  It is clear that for $L$ layers, the denominator of (\ref{eqn:prec}) can be calculated using
\begin{align}
 \del_{\phi}\langle J_{y}\rangle_{\ket{\psi_{\phi}^{(L)}}} \Bigg\vert_{\phi=0} 
= L\del_{\phi}\langle J_{y}\rangle_{\ket{\psi_{\phi}^{(1)}}}\Bigg\vert_{\phi=0} .
\label{eqn:ggg}
\end{align}

\begin{figure*}[t!]
    \begin{center}
    \includegraphics[scale=0.6]{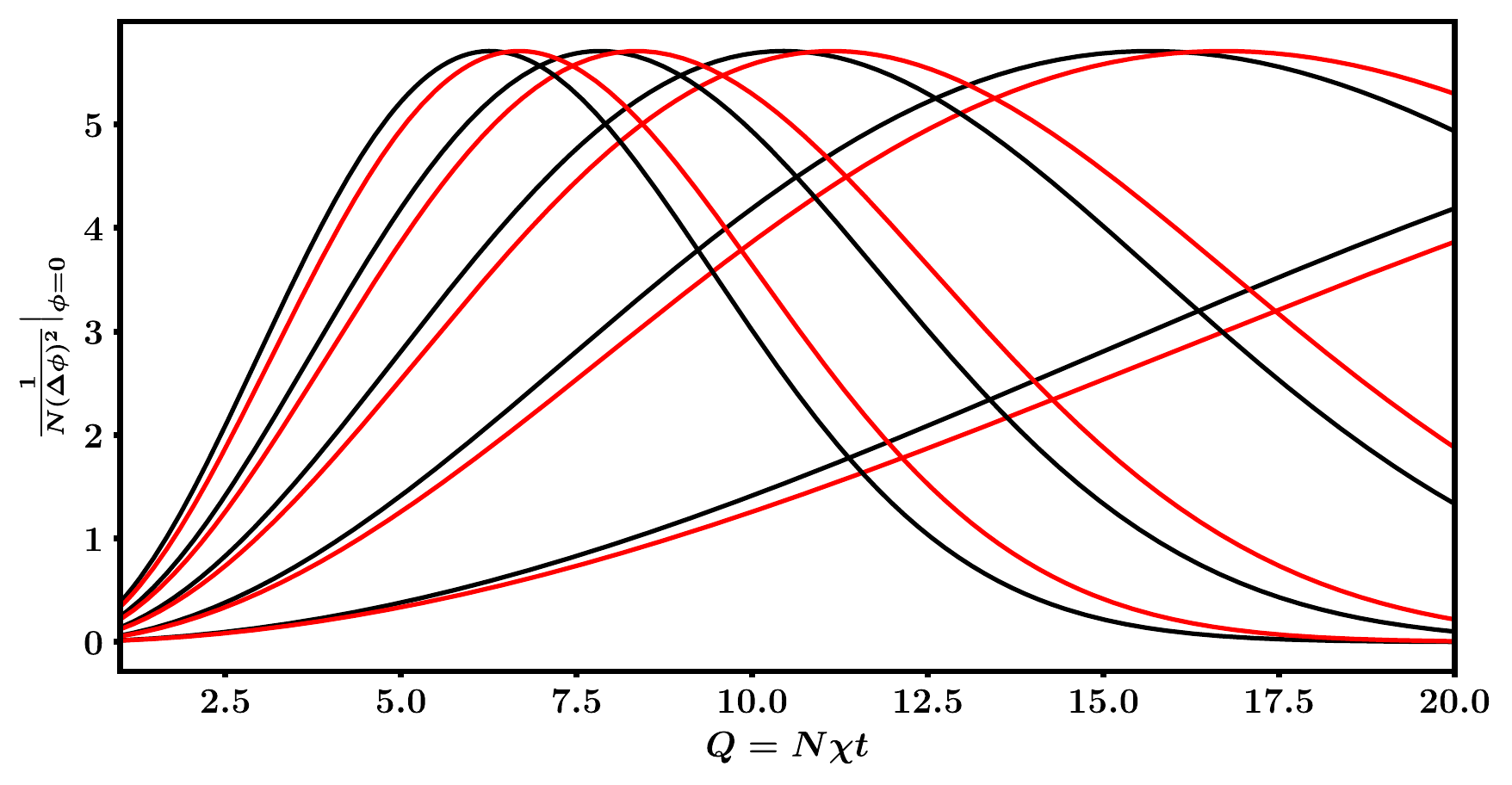}
    \caption{Black curves: inverse normalized empirical error for estimation of $\phi$ for the one-axis twist-untwist protocol in the spin wave subspace (\ref{eqn:protf}) with $N=16$ and $K=1,\ldots,5$. Red curves: the same except with $\tilde{H}_{K}$ replaced by $\tilde{H}_{K}^{(2)}$ in (\ref{eqn:protf}). For both cases, larger $Q_{\text{max}}$ corresponds to smaller $K$. \label{fig:ooo}}
    \end{center}
\end{figure*}

The numerator of (\ref{eqn:prec}) at $\phi=0$ is invariant with respect to the number of layers $L$. Because a probe state consisting of $L$ independent copies of (\ref{eqn:iii}) would have variance of the total $y$-spin component increased by a factor of $L$, we conclude that an $L$ layer protocol (\ref{eqn:lla}) allows the value of $(\Delta\phi)^{2}$ to be decreased by a factor of $L^{-1}$ compared to the protocol consisting of $L$ independent copies of (\ref{eqn:iii}).   As an alternative to the layered protocol (\ref{eqn:lla}) one may consider the parameterized state $\ket{\psi_{\phi}}=e^{i\chi t C}e^{-iL\phi J_{y}}e^{-i\chi t C}\ket{+}^{\otimes N}$, which allows to obtain an $L^{2}$ increase in the derivative of the signal, similarly as in (\ref{eqn:ggg}). The difficulty with this proposal is that the map $e^{-i\theta J_{y}}\ket{\psi} \mapsto e^{-iL\theta J_{y}}\ket{\psi}$ for all $\theta$ and $\ket{\psi}$ cannot be carried out unitarily. A proof of this fact can be obtained which is similar to proofs of ``no-go'' theorems for noiseless parametric amplification in the continuous-variable setting (i.e., that the map $\ket{\alpha}\mapsto \ket{L\alpha}$ cannot be achieved by a unitary operation).

In fact, the multiplicative improvement obtained from the $L$ layer twist-untwist protocol extends to the QFI.

\begin{proposition}
The quantum Fisher information $\mathrm{QFI}(\psi^{(L)}_{\phi})$ at $\phi=0$ for $L$ layer twist-untwist protocol (\ref{eqn:lla}) is given by
\begin{equation}
L^{2}\mathrm{QFI}(\psi^{(1)}_{\phi}) +2L(L-1)N\cos^{N-1}\chi t + (L-1)^{2}N.
\end{equation}
\label{prop:ttt}
\end{proposition}
\begin{proof}
Let $C:=J_{z}^{2}$. The symmetric logarithmic derivative at $\phi=0$ for the one layer protocol is given by
\begin{small}
\begin{equation}
\mathcal{L}^{(1)}_{\phi=0}=-2ie^{iC}J_{y}e^{-iC}\ket{+}\bra{+}^{\otimes N} + 2i \ket{+}\bra{+}^{\otimes N}e^{iC}J_{y}e^{-iC}
\label{eqn:sld}
\end{equation}
\end{small}
from which it follows that 
\begin{small}
\begin{align}
\mathcal{L}^{(L)}_{\phi=0}&=L\mathcal{L}^{(1)}_{\phi=0}\nonumber \\
&+\left( -2i(L-1)e^{iC}J_{y}e^{-iC}\ket{+}\bra{+}^{\otimes N} +h.c.\right).
\end{align}
\end{small}
One then calculates
\begin{small}
\begin{align}
\mathrm{QFI}(\psi^{(L)}_{\phi})&:= \bra{+}^{\otimes N} (\mathcal{L}^{(L)}_{\phi=0})^{2} \ket{+}^{\otimes N} \nonumber \\
% &= L^{2}\bra{+}^{\otimes N} (\mathcal{L}^{(1)}_{\phi=0})^{2} \ket{+}^{\otimes N} + 4(L-1)^{2}\bra{+}^{\otimes N} J_{y}^{2} \ket{+}^{\otimes N} \nonumber \\
% &+ 4L(L-1)\left( \bra{+}^{\otimes N} e^{iC}J_{y}e^{-iC}J_{y} \ket{+}^{\otimes N} +c.c \right)\nonumber \\
&= L^{2}\mathrm{QFI}^{(1)}(\chi,t) + N(L-1)^{2} \nonumber \\
&{}+ 4L(L-1)\left( \bra{+}^{\otimes N} e^{iC}J_{y}e^{-iC}J_{y} \ket{+}^{\otimes N} +c.c \right)
\end{align}
\end{small}
and the last term can be evaluated explicitly by using $e^{iaJ_{z}^{2}}J_{y}e^{-iaJ_{z}^{2}} = {1\over 2i}J_{+}e^{2ia(J_{z}+{1\over 2})} + h.c.$
\end{proof}

A generalization of (\ref{eqn:lla}) that allows one-axis twisting to alternate with rotations is given by
\begin{equation}
    \ket{\psi_{\phi}^{(L)}(\vec{a})}=e^{i\phi J_{y}} \prod_{j=1}^{L} e^{-i\phi J_{y}} e^{ia_{j,2} C}e^{-i\phi J_{y}}e^{ia_{j,1} C}  \ket{+}^{\otimes N}
\label{eqn:lla2}
\end{equation}
where $\vec{a}:=(a_{L,2},a_{L,1},a_{L-1,2},a_{L-1,1},\ldots)$ is the row vector of parameters. Defining the partial sums $\varphi_{\ell}=\sum_{j=1}^{\ell}\vec{a}_{j}$, where $\vec{a}_{j}$ is the $j$-th element of $\vec{a}$, allows one to evaluate \begin{align}
    N^{2}(\Delta \phi)^{2}\big\vert_{\phi =0}&={A(\varphi_{2L})\over B(\lbrace \varphi_{j}\rbrace_{j=1}^{2L})}\nonumber \\
    A(\varphi_{2L})&= 2(N+1)-2(N-1)\cos^{N-2}(2\varphi_{2L}) \nonumber \\
    B(\lbrace \varphi_{j}\rbrace_{j=1}^{2L}) &= N(N-1)^{2}\left( \sum_{j=1}^{2L-1}\sin \varphi_{j}\left( \cos^{N-2}\varphi_{j}  \right.\right. \nonumber \\
    &{} \left. \left. + \cos^{N-2}(2\varphi_{2L}-\varphi_{j}) \right)\vphantom{\sum_{j=1}^{2L-1}}\right)^{2}
    \label{eqn:prec2}
\end{align}
Solving for extrema of (\ref{eqn:prec2}) without restriction of the protocol (\ref{eqn:lla2}) is complicated, even when the method in Theorem \ref{th:1} is used. However, constraining each layer in (\ref{eqn:lla2}) to satisfy the twist-untwist condition $a_{\ell,1}=-a_{\ell,2}$ for all layers $\ell$ results in a simplification of the partial sums, viz., $\varphi_{2\ell}=0$ and $\varphi_{2\ell+1}=a_{L-\ell,2}$. One then finds that (\ref{eqn:prec2}) is minimized for equal strength twist-untwist in each layer, i.e., $\varphi_{2\ell+1}=\tan^{-1}{1\over \sqrt{N-2}}$.

A similar result as in Proposition \ref{prop:ttt} can be obtained for the $L$ layer protocol with $2L$ calls to the parameter $\phi$
\begin{equation}
\ket{\zeta_{\phi}^{(L)}}:= \left( e^{i\phi J_{y}}e^{itJ_{z}^{2}}e^{-i\phi J_{y}}e^{-itJ_{z}^{2}}\right)^{L}\vert +\rangle^{\otimes N}
\end{equation}
which generalizes the protocol $\ket{\zeta_{\phi}}$ introduced as the end of Section \ref{sec:ao}. In fact, not only does (\ref{eqn:ggg}) hold in this case, but also the symmetric logarithmic derivative obeys $\mathcal{L}^{(L)}_{\phi=0}=L\mathcal{L}^{(1)}_{\phi=0}$ which leads to the simpler result $\text{QFI}(\zeta_{\phi}^{(L)})=L^{2}\text{QFI}(\zeta_{\phi}^{(1)})$. The general layered protocol (\ref{eqn:lla2}) can also be analyzed in such a way that each layer has no net rotation.

\section{Finite range twist-untwist protocols\label{sec:fr}}
Rydberg-Rydberg atom interactions between trapped neutral atoms provide a platform for scalable many-body entanglement generation via spin-squeezing \cite{PhysRevLett.112.103601}.
For $N$ distinguishable atoms on a one-dimensional lattice with periodic boundary condition, a finite range, two-local Hamiltonian generalizing the one-axis twisting generator $J_{z}^{2}$ is given by
\begin{align}
H_{K}&={1\over 4} \sum_{j=0}^{N-1}\sum_{\substack{  i=  j-K  \text{ mod } N \\ i\neq j}}^{ j+K  \text{ mod } N } V_{j,i} Z_{j}\otimes Z_{i} 
%\nonumber \\
%&= {1\over 2}\sum_{j=1}^{N}\sum_{\substack{  i=  j-K  \text{ mod } N \\ i\neq j}}^{ j+K  \text{ mod } N } V_{j,i}\left( {Z_{j}- Z_{i}\over 2} \right)^{2}
\label{eqn:h1}
\end{align}
where $K$ is an integer in $[1,{N\over 2})$ representing the interaction range, and where $Z_{j}$ denotes the Pauli $Z$ matrix acting on atom $j$ and the identity operator on all other atoms.

A twist-untwist protocol with respect to $H_{K}$ is defined as in (\ref{eqn:iii}) with the substitution $J_{z}^{2}\rightarrow H_{K}$. For such a protocol, the expression for the denominator of (\ref{eqn:prec}) can be calculated analytically by using the identity
\begin{align}
&{}e^{i\chi t H_{K}}(\sigma_{+})_{r}e^{-i\chi t H_{K}} \nonumber \\
&{} =(\sigma_{+})_{r} \exp \left[ {i\chi t \sum_{\substack{ j=r-K \text{ mod } N \\ j\neq r }}^{ r+K \text{ mod } N }V_{r,j}Z_{j}} \right] .
\end{align}
where $(\sigma_{+})_{r}$ denotes the $\sigma_{+}$ matrix acting on atom $r$ and the identity operator on all other atoms.
The result is
\begin{small}
\begin{align}&{}\del_{\phi}\langle \psi_{\phi}\vert J_{y}\vert \psi_{\phi}\rangle\big\vert_{\phi=0} = \nonumber \\
&{} {1\over 2}\sum_{r=0}^{N-1} \sum_{\substack{\ell = -K \\ \ell \neq 0}}^{K}\left[ \sin \left( \chi  V_{r,r-\ell}t \right) \prod_{\substack{ i=-K  \\ i\neq \ell }}^{ K }\cos \left( \chi  V_{r,r-i}t \right) \right]
\label{eqn:denom}
\end{align}
\end{small}
where in the above equation, $r-\ell$ and $r-i$ are modulo $N$.

We now consider the case of a constant, finite range interaction  \begin{equation}V_{i,j}=\begin{cases} 1 & 0<\vert i-j \vert \le K \\ 0 &\vert i-j \vert > K \text{ or } i=j \end{cases},\label{eqn:knng}\end{equation}
where the inequalities are interpreted modulo $N$. In this case, $V_{i,j}$ is the adjacency matrix of the $K$-nearest neighbors graph. For this choice of $V_{i,j}$, the expression (\ref{eqn:denom}) becomes $NK\sin \chi t \cos^{2K-1}\chi t$, which obtains a maximum at $\chi t=\tan^{-1}\sqrt{1\over 2K-1}$. At this maximum, one obtains the minimal value of $(\Delta \phi)^{2}\vert_{\phi=0}={1\over 2NK\left(1-{1\over 2K}\right)^{2K-1}}$.  For $K=rN$ with ratio $r$, the $N\rightarrow \infty$ Heisenberg scaling is given by as $O(N^{-2})$ with coefficient $e\over 2r$. Further, we obtain an analytical formula for (\ref{eqn:prec}) for the general protocol \begin{equation}\ket{\psi_{\phi}}=e^{ia_{2}H_{K}}e^{-i\phi J_{y}}e^{ia_{1}H_{K}}\ket{+}^{\otimes N},
\label{eqn:fgfg}
\end{equation}
restricting to the regime of short-range interactions $K\le {N-2\over 4}$. The formula is given in Appendix \ref{sec:app1}.

In analogy with Proposition \ref{prop:aaa}, one can use the aforementioned formulas (\ref{eqn:prec4}) to show that the ratio  $\text{QFI}^{-1}/(\Delta \phi)^{2}\vert_{\phi=0}$ evaluated at the critical interaction time $\chi t = \tan^{-1}{1\over \sqrt{2K-1}}$ asymptotes to $(e+e^{-1}-2)^{-1}\approx 0.92$ as $K\rightarrow \infty$. Therefore, the quantum Cram\'{e}r-Rao bound is asymptotically only about 8\% lower that the empirical error in (\ref{eqn:prec}). At the critical interaction time, From the symmetric logarithmic derivative for the state manifold $\ket{\psi_{\phi}}$ the large $K$ asymptotic of $\text{QFI}(\psi_{\phi})$ is (see Appendix \ref{sec:app2} for the calculation)
\begin{align}
\text{QFI}(\psi_{\phi}) &= 4\text{Var}_{e^{\pm itH_{K}}\ket{+}^{\otimes N}}J_{y}\nonumber \\
&\sim 2NK(1-e^{-1})^{2},
\end{align}
which, for $K=rN$ with ratio $r$, implies an $N\rightarrow \infty$ asymptotic of $2rN^{2}(1-e^{-1})^{2}$. Therefore, both the empirical error (\ref{eqn:prec}) and the QFI exhibit Heisenberg scaling at the critical interaction time.

The formulas (\ref{eqn:prec4}) allow one to prove asymptotic optimality for finite range one-axis twist-untwist protocols among the protocols (\ref{eqn:fgfg}). The detailed statement of asymptotic optimality is the statement of the following theorem:
\begin{theorem}
Let $(\Delta \phi)^{2}\big\vert_{\phi=0}$ be defined by (\ref{eqn:prec}), with $\ket{\psi_{\phi}}$ the parameterized state prepared by protocol (\ref{eqn:fgfg}). Further, let $K=r(N-2)$ for a fixed locality parameter $0<r\le {1\over 4}$ and let $a_{j}={c_{j}\over K^{\alpha}}$ , $\alpha \ge 1/2$, for nonzero constants $c_{j}$, $j=1,2$. Then, in the limit $N\rightarrow \infty$ the unique minimum of $NK(\Delta \phi)^{2}\big\vert_{\phi=0}$ occurs at $c_{2}=-c_{1}={1\over \sqrt{2}}$. 
\label{th:2}
\end{theorem}

\begin{proof}
Call $f(a_{1},a_{2}) :=NK(\Delta \phi)^{2}\big\vert_{\phi=0}$ and note that we restrict to $a_{1}<0$ as in Theorem \ref{th:1}. The condition $K=r(N-2)$ for a fixed $0<r\le {1\over 4}$ allows to apply formula (\ref{eqn:prec4}) to calculate  $NK(\Delta \phi)^{2}\big\vert_{\phi=0}$. The factor of $NK$ in the definition of $f$ is so that the $N\rightarrow \infty$ limit of $\nabla f$ is not zero pointwise; without this factor, $f$ is asymptotically zero. Consider first $\alpha>1/2$. In order to consider the $N\rightarrow \infty$ asymptotics of (\ref{eqn:prec4}) for $K=r(N-2)$, it is useful to consider the $K\rightarrow \infty$ asymptotic, then impose the scale factor $r$ at the end.  From (\ref{eqn:anal2}) and the fact that $\cos^{m}{x\over K^{\alpha}}\sim 1$ for $\alpha >1/2$ and large $K$, it follows that $f(a_{1},a_{2})\sim {K^{2\alpha -1}\over c_{2}^{2}}$. Therefore, for $\alpha >1/2$, $f$ is arbitrarily large as $K$ increases. 

For $\alpha=1/2$, if $c_{2}\neq -c_{1}$, then the numerator of $f$ scales as $N^{2}K^{2}$ for large $K$ whereas the denominator scales as $N^{2}K$ for large $K$ (actually, the numerator of $f$ scales only as $N^{2}K$ if $a_{2}=-a_{1}+{(2\ell +1)\pi\over 2}$ for integer $\ell$, but in this situation at least one of the $a_{j}$ do not vanish with $K$ and, in particular, does not correspond to the form $a_{j}={c_{j}\over \sqrt{K}}$). Therefore, if $\alpha = 1/2$ and $c_{2}\neq -c_{1}$, then $f$ is arbitrarily large as $K$ increases. However, if $\alpha=1/2$ and $c_{2}= -c_{1}$, then $f\sim {1\over 4c_{2}^{2}e^{-2c_{2}^{2}}}$, which does not scale with $K$. The asymptotic is minimized when $c_{2}={1\over \sqrt{2}}$. Taking $K=r(N-2)$ converts the $K\rightarrow \infty$ asymptotic to $N\rightarrow \infty$ asymptotic.
\end{proof}

Note that in Theorem \ref{th:2} we have taken $K$ to be a function of $N$. Unlike the $N^{-1/2}$ decay of $a_{1(2)}$ in the case of asymptotically optimal twist-untwist protocols for full range interactions in Theorem \ref{th:1}, a $K^{-1/2}$ decay is observed in the case of range $K$ interactions. As a consequence, the rate of convergence of the optimal protocol of the form (\ref{eqn:fgfg}) to a finite range one-axis twist-untwist protocol is controlled by the locality parameter $r$. Similar to Theorem \ref{th:1}, it follows as a corollary of Theorem \ref{th:2} that the finite range one-axis twist-untwist protocol with $a_{1}\sim -{1\over \sqrt{2K}}$, $a_{2}\sim {1\over \sqrt{2K}}$ is asymptotically optimal among protocols of the form (\ref{eqn:fgfg}) restricted to weak nonlinearities $a_{j}={c_{j}\over K^{\alpha}}$, $\alpha\ge 1/2$. This is due to the fact that for all cases except $\alpha=1/2$, $c_{2}=-c_{1}$, $(\Delta \phi)^{-2}\vert_{\phi=0} = o(NK)$. However, for $\alpha=1/2$ and $c_{2}=-c_{1}$, $(\Delta \phi)^{-2}\vert_{\phi=0} = O(NK)$ with the maximum attained for $c_{2}=-c_{1}={1\over \sqrt{2}}$.
 
It is also possible to obtain formulas analogous to (\ref{eqn:prec4}) in dimension $D$ when the size of the interaction region scales as $O(r^{D}N^{D})$.  For instance, in the case of an $N\times N$ square lattice ($D=2$) with periodic boundary condition and with constant interaction on squares of $K\times K$ sites (with $K$ odd and $K\le {N-2\over 4}$), one finds that for finite range twist-untwist protocol $\ket{\psi_{\phi}}=e^{i\chi t H_{K}}e^{-i\phi J_{y}}e^{-i\chi t H_{K}}\ket{+}^{\otimes N^{2}}$ , $(\Delta \phi)^{-2}\vert_{\phi=0}$ exhibits $N^{2}K^{2}$ scaling for $\chi t\sim {1\over K}$. Further, the quantum Fisher information of $O(N^{2}K^{2})$ is obtained for $\ket{\psi_{\phi}}$ at this optimal interaction time scaling (with $K=rN$ for some ratio $r$, this expresses Heisenberg scaling for $N^{2}$ spins). To analyze the optimality of the protocol $\ket{\psi_{\phi}}$ among protocols of the form $e^{ia_{2} H_{K}}e^{-i\phi J_{y}}e^{ia_{1} H_{K}}\ket{+}^{\otimes N^{2}}$ with weak interaction strengths $a_{j}={c_{j}\over K^{2\alpha}}$, $\alpha \ge 1/2$, one considers the function $N^{2}K^{2}(\Delta \phi)^{2}\vert_{\phi=0}$ and finds that it has a unique asymptotic global minimum at $a_{2}=-a_{1}={1\over K}$ when $\alpha =1/2$.

\subsection{Translation-invariant finite range twist-untwist protocols\label{sec:gen}}
 We now consider physical systems that exhibit Bose symmetry as in Section \ref{sec:ao}, but have spatial interactions inherited from the system of distinguishable spins in Section \ref{sec:fr}. Specifically, we consider the physical scenario of two-level bosonic atoms in a ring-shaped optical lattice of $N$ sites. The internal states of the bosons are assumed to interact pairwise over a distance of $K$ sites.  The two body interaction can be written as a Heisenberg model by using the boson operators $a$ and $b$ for the internal states:
 \begin{equation}
     {1\over 4}\sum_{i=0}^{N-1}\sum_{\substack{  i=  j-K  \text{ mod } N \\ i\neq j}}^{ j+K  \text{ mod } N }  V_{j,i} \left( a^{\dagger}_{j}a_{j} -b_{j}^{\dagger}b_{j}\right) \left( a^{\dagger}_{i}a_{i} -b_{i}^{\dagger}b_{i}\right).
     \label{eqn:hp}
 \end{equation}  We assume that $a^{\dagger}_{j}a_{j} +b_{j}^{\dagger}b_{j}=1$ for all $j$, so that the sites have unit occupancy. When restricted to this subspace, it is clear that (\ref{eqn:hp}) is equal to $H_{K}$ in (\ref{eqn:h1}). However, it is possible to further restrict the interaction Hamiltonian (\ref{eqn:hp}) to describe dynamics in the translation invariant subspace of spin waves, i.e., the states spanned by the basis
 \begin{align}
     &{}\lbrace \ket{0,1,\ldots,0,1} \rbrace \cup \nonumber \\
     &{} \lbrace \sum_{\substack{i_{M}<\ldots <i_{1}\\i_{\ell}=0,\ldots,N-1}}\prod_{\ell=1}^{M}a_{i_{\ell}}^{\dagger}b_{i_{\ell}} \ket{0,1,\ldots,0,1} \rbrace_{M=1}^{N}
     \label{eqn:swb}
 \end{align}
 where $M$ is the excitation number of the spin wave and $\ket{n_{a,0},n_{b,0},\ldots,n_{a,N-1},n_{b,N-1}}$ is an insulating state with $a_{j}^{\dagger}a_{j}=n_{a,j}$ $b_{j}^{\dagger}b_{j}=n_{b,j}$.
 In terms of matrix elements, the Hamiltonian (\ref{eqn:hp}) after the projection to the spin wave subspace is equal to the Hamiltonian $\tilde{H}_{K}:= P_{B}H_{K}P_{B}$, where $P_{B}$ is the projection of  $(\mathbb{C}^{2})^{\otimes N}$ to the symmetric subspace. For simplicity, we again restrict to the constant interaction potential $V_{i,j}=1$. In the special case of $N\equiv 1\text{ mod }2$, $K={N-1\over 2}$, and $V_{i,j}=1$, then (\ref{eqn:h1}) is equal to $J_{z}^{2}-{N\over 4}$, so no projection is needed.  For the general case, one finds that the matrix elements of $\tilde{H}_{K}$ in the orthonormal basis of Dicke states are given by
\begin{align}
\langle \psi_{n}\vert \tilde{H}_{K}\vert \psi_{n'}\rangle &= {\delta_{n,n'}\over {N\choose n}}\sum_{x:\text{Ham}(x)=n}\langle x \vert H_{K}\vert x\rangle  \nonumber \\
\ket{\psi_{n}}&:= {1\over \sqrt{N\choose n}}\sum_{x:\text{Ham}(x)=n}\ket{x} \, , \, n=0,\ldots,N
\end{align}
where $\ket{x}$ is a state in the computational basis and $\text{Ham}(x)$ is its associated Hamming weight. The twist-untwist protocol in this subspace is defined by the parameterized state
\begin{equation}
\ket{\psi_{\phi}}:=e^{i\chi t \tilde{H}_{K}}e^{-i\phi J_{y}}e^{-i\chi t \tilde{H}_{K}}\ket{+}^{\otimes N}
\label{eqn:protf}
\end{equation}
where the state $\ket{+}^{\otimes N}$ corresponds to a superposition of the spin wave basis states in (\ref{eqn:swb}).

Because the exact computation of the matrix elements of $\tilde{H}_{K}$ takes exponential time in $N$, it is useful to define a model bosonic Hamiltonian for $\tilde{H}_{K}$ that helps to analyze the metrological gain obtained in the protocol (\ref{eqn:protf}). Specifically, consider the Hamiltonian 
\begin{align}
\tilde{H}_{K}^{(2)}&={1\over 4} \sum_{j=0}^{N-1}\sum_{\substack{  i=  j-K  \text{ mod } N \\ i\neq j}}^{ j+K  \text{ mod } N }  V_{j,i} P_{B}Z_{j}P_{B}Z_{i}P_{B}.
\label{eqn:h2}
\end{align}
The properties of $\tilde{H}_{K}^{(2)}$ that make it useful as a model interaction for $\tilde{H}_{K}$ are established in the following proposition.
\begin{proposition}
If $V_{i,j}=1$, then $\tilde{H}_{K}^{(2)}={2KJ_{z}^{2}\over N}$. If $V_{i,j}>0$ for all $i,j$, then $\Vert \tilde{H}_{K}^{(2)}\Vert = \Vert\tilde{H}_{K} \Vert$.
\end{proposition}
\begin{proof}
First, note that $P_{B}Z_{j}P_{B}={1\over N}\sum_{i=1}^{N}Z_{i}$, which is proven by considering computational basis states $\ket{x}, \ket{x'}$: if $\text{Ham}(x)\neq \text{Ham}(x')$ then $\langle x'\vert P_{B}Z_{j}P_{B}\vert x\rangle =0$; if  $\text{Ham}(x)= \text{Ham}(x')=n$ then 
\begin{align}
\langle x'\vert P_{B}Z_{j}P_{B}\vert x\rangle &= {1\over {N\choose n}}\sum_{x: \text{Ham}(x)=n}(-1)^{x_{j}} \nonumber \\
&= {{N-1\choose n} - {N-1\choose n-1} \over {N\choose n}} \nonumber \\
&={1-{2K\over N}} \nonumber \\
&=\langle N-n,n\vert {2J_{z}\over N}\vert N-n,n\rangle \label{eqn:ooo} ,
\end{align} and ${2J_{z}\over N}={1\over N}\sum_{i=0}^{N-1}Z_{i}$, which proves the first statement. Note that the numerator of the second line of (\ref{eqn:ooo}) is the difference between the number of ways for $n$ ones ($n-1$ ones) to appear in $x$ given that $x_{j}=0$ ($x_{j}=1$). To prove the second statement, note that for $V_{i,j}>0$,
\begin{align}
\Vert \tilde{H}_{K}^{(2)}\Vert &=\left( \ket{0}^{\otimes N},\tilde{H}_{K}^{(2)}\ket{0}^{\otimes N} \right) \nonumber \\
&=\left( \ket{0}^{\otimes N},\tilde{H}_{K}\ket{0}^{\otimes N} \right)\nonumber \\
&=  \Vert \tilde{H}\Vert.
\end{align}

\end{proof}

Numerical computation of $(\Delta \phi)^{2}\vert_{\phi=0}$ for the protocol (\ref{eqn:protf}) using both $\tilde{H}_{K}$ and $\tilde{H}^{(2)}_{K}$ for $N=16$ is shown in Fig.\ref{fig:ooo}. The $K^{-1/2}$ decay of the optimal twist-untwist parameters which appeared in the analysis of the protocol (\ref{eqn:fgfg}) for distinguishable two-level atoms is not observed. Instead, we conclude that the translation invariant twist-untwist protocol (\ref{eqn:protf}) effectively converts a finite interaction range to a multiplicative factor in the interaction strength of the one axis twisting Hamiltonian of Eq.(\ref{eqn:iii}). This results in longer interaction times required to reach maximal sensitivity for short-range bosonic twist-untwist protocols, but independence of the maximal sensitivity on the interaction range. 

\section{Discussion}

Our theorems show that the twist-untwist protocol (\ref{eqn:iii}), and its generalization to constant, finite range one-axis twisting generators, is asymptotically optimal among protocols that apply two calls to asymptotically weak one-axis twisting evolutions separated by a call to the rotation parameter of interest. We expect that our proof methods also allow one to obtain analogous results for more general spin squeezing interactions, e.g., two-axis twisting, or twist-and-turn generators \cite{PhysRevA.97.053618}. Some aspects of the present results indicate directions for future research. Firstly, the assumption that the initial Bose-Einstein condensed state $\ket{+}^{\otimes N}$ can be generated with fixed particle number $N$ is an experimentally demanding one. Therefore, analyses of twist-untwist protocols in the presence of noise could lead to more general statements of optimality. Secondly, the weak nonlinearity constraints in Sections \ref{sec:ao}, \ref{sec:fr} are well-motivated because the respective values of $(\Delta\phi)^{2}\vert_{\phi=0}$ for the twist-untwist protocols are minimized and exhibit Heisenberg scaling for such interaction times. It would be useful to design more general quantum metrology protocols that: 1. subsume the present twist-untwist protocol at weak nonlinearities, but can be further extended to nearly saturate the QFI for all interaction times, 2. still make use of a total spin readout. The present work provides a foundation for future analyses of imperfect or noisy, generalized twist-untwist protocols for atom-based sensing.

\acknowledgements

The authors were supported by the Laboratory Directed Research and Development (LDRD) program of Los Alamos National Laboratory (LANL) under project number 20210116DR. Michael J. Martin was also supported by supported by the U.S. Department of Energy, Office of Science, National Quantum Information Science Research Centers, Quantum Science Center. Los Alamos National Laboratory is managed by Triad National Security, LLC, for the National Nuclear Security Administration of the U.S. Department of Energy under Contract No. 89233218CNA000001.

\bibliography{phasebib.bib}

\onecolumngrid
\appendix
\section{\label{sec:app1}Formula for (\ref{eqn:prec}) for finite range one axis twist-untwist protocol}
The formulas for the numerator and denominator of (\ref{eqn:prec}) that allow to prove Theorem \ref{th:2} are given by
\begin{align}
    \text{Var}_{\ket{\psi_{\phi}}}J_{y}\big\vert_{\phi=0}
    &= {1\over 2} \left(J_{+}\ket{\psi_{\phi}},J_{+}\ket{\psi_{\phi}} \right) \big\vert_{\phi=0} -{1\over 2}\text{Re}\left(J_{-}\ket{\psi_{\phi}},J_{+}\ket{\psi_{\phi}} \right) \big\vert_{\phi=0} \nonumber \\ &= {N\over 4}\left( 1+\sum_{j=0}^{K-1}\Bigg[ \left(1-\cos^{K+j-1}(2a_{1}+2a_{2})\right)\cos^{2K-2j}(a_{1}+a_{2})  \right. \nonumber \\
    &{} \hspace{1cm}+ \left.  \left(1-\cos^{j+1}(2a_{1}+2a_{2})\right)\cos^{4K-2j-2}(a_{1}+a_{2})\Bigg] \vphantom{\sum_{j=0}^{K-1}}\right) \nonumber \\
    \del_{\phi}\langle \psi_{\phi}\vert J_{y} \vert \psi_{\phi}\rangle\big\vert_{\phi=0} &= -2\text{Im} \left( e^{ia_{1}H_{K}}J_{y}e^{ia_{2}H_{K}}\ket{+}^{\otimes N},J_{y}e^{i(a_{1}+a_{2})H_{K}}\ket{+}^{\otimes N} \right) \nonumber \\
    &= {N\over 2}\sin a_{2}\sum_{j=0}^{K-1}\left( \cos^{K+j-1}a_{2} + \cos^{K+j-1}(2a_{1}+a_{2}) \right)(\cos a_{1}\cos(a_{1}+a_{2}))^{K-j}
    \label{eqn:prec4}
    \end{align}
    which are valid for $1\le K \le {N-2\over 4}$. The upper restriction on $K$ is imposed because we have assumed periodic boundary conditions.
    
    \section{\label{sec:app2}$K\rightarrow \infty$ asymptotics for $\text{QFI}(\psi_{\phi})$ at optimal interaction time}
    To calculate $\text{QFI}(\psi_{\phi}):= 4\text{Var}_{e^{-itH_{K}}\ket{+}^{\otimes N}}J_{y}$ at $t=\tan^{-1}{1\over \sqrt{2K-1}}$, apply the formula $1-\cos^{\ell}(2\tan^{-1}{1\over \sqrt{2K-1}}) = 1-\left( {K-1\over K}\right)^{\ell}$ for $\ell \in \mathbb{Z}_{\ge 0}$ to the formula
    \begin{align}
    4\text{Var}_{e^{-itH_{K}}\ket{+}^{\otimes N}}J_{y}&= N\left( 1+\sum_{j=1}^{K}\left[ \vphantom{\sum_{j=1}^{K}} (1-\cos^{j+1}(2t))\cos^{4K-2j-2}(t) \right. \right.\nonumber \\
    &{} \left. \left. +(1-\cos^{K+j-1}(2t))\cos^{2K-2j}t \vphantom{\sum_{j=1}^{K}}\right] \right).
    \end{align}
    Summing the geometric series gives that the right hand side of the above equation is equal to $N(1+g_{1}(K)+g_{2}(K))$ with
    \begin{align}
    g_{1}(K)&:= (2K-1)(1-{1\over 2K})^{K}\left( \left( {2K\over 2K-1} \right)^{K}-1\right) - \left( {2K^{2}-3K+1\over 2K^{2}}\right)^{K}{(2K-1)K\over K-1}\left( 1-\left( {2(K-1)\over 2K-1} \right)^{K} \right)\nonumber \\
    &\sim 2K(e^{-1}-e^{-1/2})(e^{-1}-e^{1/2}) \nonumber \\
    g_{2}(K)&:= 2K\left( 1-{1\over 2K}\right)^{2K} \left( \left( \left({2K\over 2K-1}\right)^{K}-1\right) -{K-1\over K}\left( 1- \left({2(K-1)\over 2K-1}\right)^{K} \right) \right)\nonumber \\
    &\sim 2K(e^{-1/2}-e^{-1})(1-e^{-1/2}).
    \end{align}
    From these asymptotics, one obtains $\text{QFI}(\psi_{\phi}) \sim 2KN(1-e^{-1})^{2}$ as stated in the main text.

\end{document}